\documentclass[11pt,a4paper]{snamsart}

\usepackage[utf8]{inputenc}
\usepackage[T1]{fontenc}

\usepackage{tikz}

\definecolor{durhampurple}{cmyk}{.51,.91,0,.34}
\definecolor{purple}{cmyk}{0.55,1,0,0.15}
\definecolor{darkblue}{cmyk}{1,0.58,0,0.21}
\usepackage[colorlinks=true,
  linkcolor=black,
  citecolor=purple,
  urlcolor=darkblue,
  pdfauthor={Andrei Krokhin and Jakub Oprsal},
  pdftitle={The complexity of 3-colouring H-colourable graphs}]%
  {hyperref}
\clearpage{}%
\newtheorem{theorem}{Theorem}[section]
\newtheorem{lemma}[theorem]{Lemma}

\newtheorem{corollary}[theorem]{Corollary}

\theoremstyle{definition}
\newtheorem{definition}[theorem]{Definition}
\newtheorem{conjecture}[theorem]{Conjecture}

\theoremstyle{remark}
\newtheorem{example}[theorem]{Example}

\newtheorem{remark}[theorem]{Remark}
\numberwithin{equation}{section}

\newcommand{\vdotsequals}{\mathrel{\setbox0=\hbox{$\equals$}\makebox[\the\wd0]{\vdots}}}

\let\rel\mathbf         %
\let\clo\mathscr        %
\let\equals\approx      %

\let\epsilon\varepsilon

\let\Union\bigcup

\DeclareMathOperator{\PCSP}{\text{PCSP}}
\DeclareMathOperator{\Pol}{Pol}

\newcommand{\NP}{\textsf{NP}}
\newcommand{\Ptime}{\textsf{P}}

\newcommand{\yes}{{\scshape yes}}
\newcommand{\no}{{\scshape no}}

\newcommand{\e}{\mathsf E}
\newcommand{\vt}{\mathsf V}
\newcommand{\edge}{\mathop{\Delta_\e}}
\newcommand{\vertex}{\mathop{\Delta_\vt}}
\let\bd\partial
\clearpage{}%

\begin{document}

\author{Andrei Krokhin}
\address{Department of Computer Science, Durham University\\ Durham, UK}
\email{andrei.krokhin@durham.ac.uk}

\author{Jakub Opršal}
\address{Department of Computer Science, Durham University\\ Durham, UK}
\email{jakub.oprsal@durham.ac.uk}

\title{The complexity of 3-colouring H-colourable graphs}

\thanks{This paper was supported by the UK EPSRC grant EP/R034516/1.}

\begin{abstract}
  We study the complexity of approximation on satisfiable instances for graph homomorphism problems.  For a~fixed graph $\rel H$, the $\rel H$-colouring problem is to decide whether a~given graph has a~homomorphism to $\rel H$. By a~result of Hell and Nešetřil, this problem is \NP-hard for any non-bipartite graph $\rel H$.  In the context of promise constraint satisfaction problems, Brakensiek and Guruswami conjectured that this hardness result extends to promise graph homomorphism as follows: fix any non-bipartite graph $\rel H$ and another graph $\rel G$ with a~homomorphism from $\rel H$ to $\rel G$, it is \NP-hard to find a~homomorphism to $\rel G$ from a~given $\rel H$-colourable graph.  Arguably, the two most important special cases of this conjecture are when $\rel H$ is fixed to be the complete graph on 3 vertices (and $\rel G$ is any graph with a~triangle) and when $\rel G$ is the complete graph on 3 vertices (and $\rel H$ is any 3-colourable graph). The former case is equivalent to the notoriously difficult approximate graph colouring problem. In this paper, we confirm the Brakensiek-Guruswami conjecture for the latter case.  Our proofs rely on a novel combination of the universal-algebraic approach to promise constraint satisfaction, that was recently developed by Barto, Bulín and the authors, with some ideas from algebraic topology. 

\end{abstract}

\maketitle
  \section{Introduction}

  In this paper we investigate the complexity of finding an approximate solution to satisfiable instances.  For example, for the problem of 3-colouring a~graph, one~natural approximation version is the \emph{approximate graph colouring} problem: The goal is to find a~$c$-colouring of a~given 3-colourable graph.
  There is a~huge gap in our understanding of the complexity of this problem. The best known efficient algorithm uses roughly $c=O(n^{0.199})$ colours where $n$ is the number of vertices of the graph \cite{KT17}. It has been long conjectured the problem is \NP-hard for any fixed constant $c\ge 3$, but the state-of-the-art here has only recently been improved from $c=4$ \cite{KLS00,GK04} to $c=5$ \cite{BKO19,BKO18}.%

  Graph colouring problems naturally generalise to graph homomorphism problems and further to constraint satisfaction problems (CSPs). In a~graph homomorphism problem, one is given two graphs and needs to decide whether there is a~homomorphism (edge-preserving map) from the first graph to the second \cite{HN04}. The CSP is generalisation of this that uses arbitrary relational structures in place of graphs. One particularly important case that attracted much attention is when the second graph/structure is fixed, this is the so-called non-uniform CSP \cite{BKW17,FV98}. For graph homomorphisms, this gives the so-called $\rel H$-\emph{colouring} problem: decide whether a~given graph has a~homomorphism to a~fixed graph $\rel H$ \cite{HN04}. The \Ptime\ vs.\ \NP-complete dichotomy of $\rel H$-colouring given in \cite{HN90} was one of the base cases that supported the Feder-Vardi dichotomy conjecture for CSPs \cite{FV98}. 
  The study of the complexity of the (standard) CSP and the final resolution of the dichotomy conjecture \cite{Bul17,Zhu17} was greatly influenced by the algebraic approach (see survey \cite{BKW17}). This approach has also made important contributions to the study of approximability of CSPs (e.g.\ \cite{BK16}).

  Brakensiek and Guruswami \cite{BG16-graph,BG18-structure} suggested that perhaps  progress on approximate graph colouring and similar open problems can be made by looking at a~broader picture, by extending it to \emph{promise graph homomorphism} and further to the \emph{promise constraint satisfaction problem} (PCSP).
  The promise graph homomorphism is an approximation version of the graph homomorphism problem in the following sense: we fix (not one but) two graphs $\rel H$ and $\rel G$ such that there is a~homomorphism from $\rel H$ to $\rel G$; the goal is then to find a~$\rel G$-colouring for a~given $\rel H$-colourable graph. The \emph{promise} is that the input graph is always $\rel H$-colourable. The general promise CSP (PCSP, for short) is a~natural generalisation of this to arbitrary relational structures.

  Given the huge success of the algebraic approach to the CSP, it is natural to investigate what it can do for PCSPs. This investigation was started by Austrin, H{\aa}stad, and Guruswami \cite{AGH17}, with an application to a~promise version of SAT. It was further developed by Brakensiek and Guruswami \cite{BG16-graph,BG18-structure,BG19} and applied to a~range of problems, including versions of approximate graph and hypergraph colouring. A~recent paper \cite{BKO19,BKO18} describes a~general abstract algebraic theory for PCSPs. However, the algebraic theory of PCSP is still very young and much remains to be done both in further developing it and in applying it to specific problems. We note that the aforementioned \NP-hardness of 5-colouring a~given 3-colourable graph was proved in \cite{BKO19,BKO18} by applying this abstract theory.

  In the present paper, we apply this general theory to prove \NP-hardness for an important class of promise graph homomorphism problems.

  \subsubsection*{Related work}

  The notion of the PCSP has been coined in \cite{AGH17}, but problems from the class has been around for a~long time, e.g.\ the approximate graph colouring \cite{GJ76}.

  Most notable examples of PCSPs studied before are related to graph and hypergraph colouring. We already mentioned some results concerning colouring 3-colourable graphs with a~constant number of colours. 
  By additionally assuming non-standard (perfect-completeness) variants of the Unique Games Conjecture, \NP-hardness was shown for all constant $c\ge 3$ \cite{DMR09}.
  Without additional complexity-theoretic assumptions, the strongest known \NP-hardness results for  colouring $k$-colourable graphs are as follows. 
  For any $k\ge 3$, it is \NP-hard to colour a~given $k$-colourable graph with $2k-1$ colours \cite{BKO19,BKO18}. For large enough $k$,  it is \NP-hard to colour a~given $k$-colourable graph with $2^{\Omega(k^{1/3})}$ colours \cite{Hua13}.
  The only earlier result about promise graph homomorphisms (with $\rel H\ne \rel G$) that involves more than graph colouring is the \NP-hardness of 3-colouring graphs that admit a~homomorphism to $\rel C_5$, the five-element cycle \cite{BKO18}, which is the simplest problem within the scope of the main result of this paper.

  A~colouring of a~hypergraph is an assignment of colours to its vertices that leaves no edge monochromatic. It is known that, for any constants $2\le k\le c$, it is \NP-hard to find a $c$-colouring of a~given 3-uniform $k$-colourable hypergraph \cite{DRS05}.
  Further variants of approximate hypergraph colouring, e.g.\ relating to strong or rainbow colourings, were studied in \cite{ABP18,BG16-graph,BG17,GL17}, but most complexity classifications related to them are still open in full generality. There are also hardness results concerning hypergraph colouring with a~super-constant number of colours, e.g.~\cite{ABP19,Bha18}.

  An accessible exposition of the algebraic approach to the CSP can be found in \cite{BKW17}, where many ideas and results leading to (but not including) the resolution \cite{Bul17,Zhu17} of the Feder-Vardi conjecture are presented. The volume \cite{KZ17} contains surveys concerning many aspects of the complexity and approximability of CSPs.

  The first link between the algebraic approach and PCSPs was found by Austrin, H\aa{}stad, and Guruswami \cite{AGH17},  where they studied a promise version of $(2k+1)$-SAT called $(2+\epsilon)$-SAT.
  It was further developed by Brakensiek and Guruswami \cite{BG16-graph,BG18-structure,BG19}. They use a~notion of \emph{polymorphism} (which is the central concept in the algebraic theory of CSP) suitable for PCSPs, and show that the complexity of a~PCSP is fully determined by its polymorphisms --- in the sense that that two PCSPs with the same set of polymorphisms have the same complexity. They also use polymorphisms to prove several hardness and tractability results. The algebraic theory of PCSP was lifted to an abstract level in \cite{BKO19,BKO18}, where it was shown that abstract properties of polymorphisms determine the complexity of PCSP. The main result of this paper heavily relies on \cite{BKO19,BKO18}.

  \subsubsection*{Our contribution}

  The approximate graph colouring problem is about finding a~$c$-colouring of a~given $3$-colourable graph. In other words, it relaxes the goal in $3$-colouring. We can instead insist that we want to find a~$3$-colouring, but strengthen the promise, i.e., fix a~$3$-colourable graph $\rel H$, and ask how hard it is to find a~$3$-colouring of a~given $\rel H$-colourable graph.
  We prove that this problem is \NP-hard for any non-bipartite graph $\rel H$ that is 3-colourable. Note that if $\rel H$ is bipartite, then this problem is solvable in polynomial time, and therefore our result completes a~dichotomy of this special case of the promise graph homomorphism problem.

  The scope of our result can be seen as a~certain dual of approximate graph colouring in the landscape of promise graph homomorphism, in the following sense.
  It is not hard to see that, in order to prove that promise graph homomorphism is \NP-hard for any pair of non-bipartite graphs $(\rel H,\rel G)$, it enough to prove this for all pairs $(\rel C_k,\rel K_n)$, $k\ge 3$ odd and $n\ge 3$, where the first graph is an odd cycle and the second is a~complete graph.
  This is because we have a~chain of homomorphisms
  \begin{equation}
      \label{eq:hom-chain}
      \ldots \to \rel C_k \to \ldots \to \rel C_5 \to \rel C_3=\rel K_3 \to \rel K_4 \to \ldots \to \rel K_n \to \ldots
  \end{equation}
  and, for each $(\rel H,\rel G)$ with a homomorphism $\rel H\rightarrow\rel G$, the problem $\PCSP(\rel H,\rel G)$ admits a (trivial) reduction from  $\PCSP(\rel C_k,\rel K_n)$ where $k$ is the size of an odd cycle in $\rel H$ and $n$ is the chromatic number of~$\rel G$ (so we have $\rel C_k\rightarrow \rel H\rightarrow \rel G\rightarrow \rel K_n$). The chain of homomorphisms (\ref{eq:hom-chain}) has a natural middle point $\rel K_3$. From this middle point, the right half of the chain corresponds to approximate graph colouring and the left half is the scope of this paper.

  The result of this paper can also be viewed as \NP-hardness of colouring $(2+\epsilon)$-colourable graphs with $3$ colours. This statement can be made more formal using so-called \emph{circular chromatic number} (see e.g.~\cite[p.\ 7]{WZ19}): one can check that indeed a~graph maps homomorphically to $\rel C_{2k+1}$ if and only if its circular chromatic number is at most $2+1/k$.

  We remark that the promise graph homomorphism problem has not been studied much beyond the case of approximate graph colouring. This somewhat narrow focus of earlier research can probably be explained by the serious difficulties encountered already in this special case. However, broadening the scope is advantageous, as this brings new methods into the picture, which can potentially resolve the difficult special case too. 

  Our proofs rely on the universal-algebraic approach to promise constraint satisfaction, that was recently developed by Barto, Bulín, and the authors \cite{BKO19,BKO18}, as well as on some ideas from algebraic topology.  To the best of our knowledge, this is the first time when ideas from universal algebra and algebraic topology are applied together to analyse the complexity of approximation.  We remark that three earlier results on the complexity of approximate hypergraph colouring \cite{ABP18,Bha18,DRS05} were based on results from topological combinatorics using the Borsuk-Ulam theorem or similar~\cite{Lov78,Mat03}. Their use of topology seems different from ours, and it remains to be seen whether they are all occurrences of a~common pattern.

  \subsubsection*{Subsequent work}

In \cite{WZ19}, Wrochna and Živný adapted and formalized the topological intuition of the present paper and generalized our result. In particular, they showed that the complexity of $\rel G$-colouring $(2+\epsilon)$-colourable graphs depends only on an inherent topology of the graph $\rel G$. They showed that this implies \NP-hardness of colouring $(2+\epsilon)$-colourable graphs with $4-\epsilon$ colours, and obtained a~similar result for `square-free' graphs $\rel G$.
Remarkably, using somewhat related methods, they also improved the state-of-the-art of the standard approximate graph colouring. More precisely, building on the results of Huang \cite{Hua13}, they proved \NP-hardness of colouring a~$k$-colourable graph with ${k \choose {\lfloor k/2\rfloor}}-1$ colours for any $k\geq 4$, which improves both \cite{Hua13} (where the number of colours is smaller, and the result is proved only only for large enough $k$) and the result of \cite{BKO19,BKO18} for any $k > 5$.

\subsubsection*{Overview of key technical ideas}

We prove the hardness via a~reduction from Gap Label Cover. The general structure of the proof is similar to \cite{AGH17}, where they first give a general sufficient condition for the existence of such a reduction for general PCSPs, and then apply it to specific PCSPs which they call $(2+\epsilon)$-SAT.
However, the structure of our problems is rather more complicated --- in particular, our problems do not satisfy the sufficient condition from \cite{AGH17}, so we need to do substantially more. It was shown in \cite{BKO18} that a~reduction in the style of \cite{AGH17} works under a~weaker structural assumption (than \cite{AGH17}), and the technical part of this paper shows that this weaker assumption is satisfied for our problems. Let us explain this in more detail.

The general reduction in \cite{AGH17} encodes an instance of Gap Label Cover as an instance of a PCSP instance by using a~polymorphism gadget.
(Roughly, polymorphisms are multivariate functions compatible with the constraint relations of the PCSP.)
That is, solutions of this gadget are polymorphisms, one for each variable of the original label cover instance. In this encoding, the arity of polymorphisms corresponds to the size of label sets in the label cover instance, so an assignment of a~label to a~variable corresponds to choosing a coordinate in the corresponding polymorphism. The completeness of the reduction follows automatically from the structure of the gadget. The proof of soundness uses the assumption such that any polymorphism of the PCSP at hand essentially depends only on a~bounded number of variables (i.e., is a~\emph{junta}) --- this is the sufficient condition.  
It is not hard to prove that this is enough to provide a~good-enough approximation for a label cover instance. One can assign to each variable of the label cover instance a label chosen uniformly at random from the bounded-size set of labels corresponding to these essential variables of the corresponding polymorphism.

This approach does not work directly for our problems, since we do not have the property that all polymorphisms are juntas.
However, we can use a stronger version of the above mentioned result from \cite{AGH17} given in \cite{BKO18}.
This stronger version weakens the assumption the all polymorphisms are juntas --- instead, we assume that we can map our polymorphisms to another set of multivariate functions that does have this property, and we can do it in a way that works well with the label cover constraints, so we can use it to identify the (bounded-size set of) important coordinates in our polymorphisms. Formally, such a~map is called a~minion homomorphism (see Definition~\ref{def:minion-homomorphism}). This notion plays a~very important role in the algebraic theory of PCSP (see \cite{BKO19,BKO18}). The construction of this map and the proof that it is a minion homomorphism is the technical content of the paper. Once this is done, our main result follows from \cite{BKO18}.

In order to identify which coordinates in polymorphisms are important and which are `noise', we need to analyse the structure of our polymorphisms. In our case, the polymorphisms are simply 3-colourings of direct powers of a~fixed odd cycle. Since $\rel K_3$ is also an odd cycle, we have that both graphs defining our PCSPs are discretisations of a~circle. Our analysis is inspired by ideas from algebraic topology. We assign to each coordinate an integer, called a~\emph{degree}. For a~unary polymorphism, i.e., a~homomorphism from the odd cycle to the~3-cycle, this degree has a~precise intuitive meaning: it is the number of times the domain cycle wraps around the range cycle under the homomorphism. This corresponds to the topological degree of a~continuous map between two copies of a~circle. We further generalise this degree to higher arity polymorphisms --- roughly, the degree at a~certain coordinate is supposed to count how many times that coordinate wraps around the circle when ignoring all other coordinates. To define this number formally and consistently, we borrow a~few notions from algebraic topology: To graphs and graph homomorphisms, we associate Abelian groups and group homomorphisms. These correspond to so-called groups of chains in topology, that are further used to define homology. However, we do not follow this theory that far, and define the degrees directly from these group homomorphisms.
Finally, we show that only a~bounded number of variables in a~polymorphism can have a~non-zero degree and use this fact to define our minion homomorphism.

\subsubsection*{Organisation of the paper}

In Section~\ref{sec:preliminaries}, we introduce technical notions that we will need in our proof.  Section~\ref{sec:main} states our main result and the result from \cite{BKO18} that our proof relies on. In Section~\ref{sec:topology}, we give an overview of the topological intuition of the proof. This section is useful for those who want to get a~deeper understanding of the topological intuition --- however, it is not required for checking the formal proof of the main  result, which is presented in Section~\ref{sec:proof}.
\section{Preliminaries}
  \label{sec:preliminaries}

\subsection{Promise graph homomorphism problems}

The approximate graph colouring problem and promise graph homomorphism problem are special cases of the PCSP, and we use the theory of PCSPs. However, we will not need the general definitions, so we define everything only for graphs. For general definitions, see, e.g.\ \cite{BKO18}.
All graphs in this paper are loopless (i.e. irreflexive).

\begin{definition}
A \emph{homomorphism} from a~graph $\rel H=(V(H),E(H))$ to another graph $\rel G=(V(G),E(G))$ is a~map $h\colon V(H)\to V(G)$ such that $(h(u),h(v))\in E(G)$ for every $(u,v)\in E(H)$.  In this case we write $h\colon \rel H\to \rel G$, and simply $\rel H\to \rel G$ to indicate that a~homomorphism exists.
\end{definition}

We now define formally the promise graph homomorphism problem.

\begin{definition}
Fix two graphs $\rel H$ and $\rel G$ such that $\rel H \rightarrow \rel G$.
\begin{itemize}
  \item The \emph{search} variant of $\PCSP(\rel H,\rel G)$ is, given an~input graph $\rel I$ that maps homomorphically to $\rel H$, \emph{find} a~homomorphism $h\colon \rel I\to \rel G$.
  \item The \emph{decision} variant of $\PCSP(\rel G,\rel H)$ requires, given an input graph $\rel I$ such that either $\rel I\to \rel H$ or $\rel I\not\to\rel G$, to output \yes{} in the former case, and \no{} in the latter case. 
\end{itemize}
\end{definition}

Note that there is an obvious reduction from the decision variant of each PCSP to the search variant, but it is not known whether the two variants are equivalent for each PCSP. The hardness results in this paper hold for the decision (and hence also for the search) version of $\PCSP(\rel H,\rel G)$. 

It is obvious that if at least one of $\rel H, \rel G$ is bipartite then the problem can be solved in polynomial time by using an algorithm for 2-colouring.

\begin{conjecture}[\cite{BG18-structure}]
  Let $\rel H$ and $\rel G$ be any non-bipartite graphs with $\rel H\to \rel G$. Then $\PCSP(\rel H,\rel G)$ is \NP-hard.
\end{conjecture}

The graphs that we will be working with in this paper are cycles and their direct powers.
As usual, we denote by $\rel K_k$ the complete graph on $k$ vertices, and by $\rel C_k$ the $k$-cycle. We will assume throughout that the set of vertices of both graphs is $\{0,1,\ldots,k-1\}$ and that the of the edges of the $k$-cycle are $(0,1)$, \dots, $(k-2,k-1)$, $(k-1,0)$. 

\begin{definition}\label{def:nth-power}
The \emph{$n$-th direct (or tensor) power} of a~graph $\rel G$ is the graph $\rel G^n$ whose vertices are all $n$-tuples of vertices of $\rel G$ (i.e., $V(G^n) = V(G)^n$), and whose edges are defined as follows: $( (u_1,\dots,u_n), (v_1,\dots,v_n) )$ is an edge of $\rel G^n$ if and only if $(u_i,v_i)$ is an edge of $\rel G$ for all $i \in \{1,\dots,n\}$.
\end{definition}

\subsection{Polymorphisms}

Although this paper does not use the general PCSPs, we will use the tools developed for analysis of these kind of problems. Namely, we use the notions of polymorphisms \cite{AGH17,BG18-structure}, minions and minion homomorphisms \cite{BKO19,BKO18}. We introduce these notions in the special case of graphs below. The general definitions and more insights can be found in \cite{BKO18,BKW17}.

\begin{definition}
An $n$-ary \emph{polymorphism} from a~graph $\rel G$ to a~graph $\rel H$ is a~homomorphism from $\rel G^n$ to $\rel H$.
To spell this out, it is a~mapping $f\colon V(G)^n \to V(H)$ such that, for all tuples $(u_1,v_1)$, \dots, $(u_n,v_n)$ of edges of $\rel G$, we have
\[
  ( f( u_1,\dots,u_n ), f( v_1,\dots,v_n ) ) \in E(H)
.\]
We denote the set of all polymorphisms from $\rel G$ to $\rel H$ by $\Pol(\rel G,\rel H)$.
\end{definition}

\begin{example} The $n$-ary polymorphisms from a~graph $\rel G$ to the $k$-clique $\rel K_k$ are the $k$-colourings of $\rel G^n$. 
\end{example}

An important notion in our analysis of polymorphisms is that of an essential coordinate.

\begin{definition}
  \label{def:essential}
A coordinate $i$ of a function $f\colon A^n\rightarrow B$ is called \emph{essential} if there exist $a_1,\ldots,a_n$ and $b_i$ in $A$ such that
\[
  f(a_1,\ldots,a_{i-1},a_i,a_{i+1},\ldots, a_n)\ne f(a_1,\ldots,a_{i-1},b_i,a_{i+1},\ldots, a_n). 
\]
A coordinate of $f$ that is not essential is called \emph{inessential} or \emph{dummy}.
\end{definition}

The set of polymorphisms between any two graphs is closed under the operation of taking a~minor, that is, it is a~minion.
Let us formally define these notions.

\begin{definition} An~$n$-ary function $f\colon A^n\to B$ is called a~\emph{minor} of an $m$-ary function $g\colon A^m \to B$ given by a~map $\pi \colon \{1,\dots,m\} \to \{1,\dots,n\}$ if
\[
  f(x_1,\dots,x_n) = g(x_{\pi(1)},\dots,x_{\pi(m)})
\]
for all $x_1,\dots,x_n \in A$.
\end{definition}

Alternatively, one can say that $f$ is a~minor of $g$ if it is obtained from $g$ by identifying variables, permuting variables, and introducing inessential variables.

\begin{definition}
  Let $\clo O(A,B) = \{f\colon A^n\rightarrow B\mid n\ge 1\}$. A~\emph{(function) minion} $\clo M$ on a pair of sets $(A,B)$ is a~non-empty subset of $\clo O(A,B)$ that is closed under taking minors. For fixed $n\ge 1$, let $\clo M^{(n)}$ denote the set of $n$-ary functions from $\clo M$.
\end{definition}

\begin{definition}\label{def:minion-homomorphism} 
  Let $\clo M$ and $\clo N$ be two minions (not necessarily on the same pairs of sets). A~mapping $\xi\colon \clo M \to \clo N$ is called a~\emph{minion homomorphism} if 
  \begin{enumerate}
    \item it preserves arities, i.e., maps $n$-ary functions to $n$-ary functions for all $n$, and
    \item it preserves taking minors, i.e., for each $\pi\colon \{1,\dots,m\} \to \{1,\dots,n\}$ and each $g\in \clo M^{(m)}$ we have
    \[
      \xi (g)(x_{\pi(1)},\dots,x_{\pi(m)}) = \xi (g(x_{\pi(1)},\dots,x_{\pi(m)}))
    .\]
  \end{enumerate}
\end{definition}

We refer to \cite[Example 2.22]{BKO18} for an example of a~minion homomorphism.

\begin{definition}
A~minion $\clo M$ is said to have \emph{essential arity at most $k$}, if each function $f\in \clo M$ has at most $k$ essential variables. It is said to have \emph{bounded essential arity} if it 
has essential arity at most $k$ for some $k$.
\end{definition}

\begin{remark}
It is well known (see, e.g. \cite{GL74}), and not hard to check, that the minion $\Pol(\rel K_3,\rel K_3)$ has essential arity at most 1.
However, it is easy to show that, for any odd $k>3$, the minion $\Pol(\rel C_k,\rel K_3)$ does not have bounded essential arity. Fix a~homomorphism $h\colon \rel C_k\to\rel K_3$ such that $h(0)=h(2)=0$ and $h(1)=1$ and define the following function from $\rel C_k^n$ to $\rel K_3$:
\[
  f(x_1,\ldots,x_n) =
    \begin{cases}
      2 & \text{if }x_1=\ldots=x_n=1,\\
      h(x_1) & \text{otherwise.} 
    \end{cases}
\]
It is easy to check that $f\in \Pol(\rel C_k,\rel K_3)$.
By using Definition~\ref{def:essential} with $a_1=\ldots=a_n=1$ and $b_i=0$, one can verify that every coordinate $i$ of $f$ is essential.
\end{remark}

Our proof will rely on the following theorem which is a special case of a result in \cite{BKO18} that generalised \cite[Theorem~4.7]{AGH17}. We remark that the proof of this theorem is by a reduction from Gap Label Cover, which is a~common source of inapproximability results.

\begin{theorem}[{\cite[Proposition 5.15]{BKO18}}] \label{thm:bounded-arity}
  Let $\rel H, \rel G$ be graphs such that $\rel H \to \rel G$. 
  Assume that there exists a minion homomorphism $\xi\colon\Pol(\rel H,\rel G) \rightarrow \clo M$ for some minion $\clo M$ on a pair of (possibly infinite) sets such that $\clo M$ has bounded essential arity and does not contain a~constant function (i.e., a~function without essential variables). Then $\PCSP(\rel H,\rel G)$ is \NP-hard.
\end{theorem}

\subsection{Graph homology}

In this section we introduce a~simple way to associate Abelian groups and group homomorphisms to graphs and graph homomorphisms. We will use this connection to find a~minion homomorphism needed to apply Theorem~\ref{thm:bounded-arity} to $\rel H = \rel C_k$, $k\ge 3$ odd, and $\rel G = \rel K_3$.  What we describe here is a~special case of standard notions in algebraic topology \cite{Hat01}, but we do not assume any topology background.

For an edge $(u,v)$ in a graph $\rel G$, let $[u,v]$ denote an orientation of the edge from $u$ to $v$.

\begin{definition}
Fix a graph $\rel G=(V(G),E(G))$. Let $\vertex(\rel G)$ denote the free Abelian group with generators $[v], v\in V(G)$. That is, the elements of this group are formal sums $\sum_{v\in V(G)}{c_v[v]}$, where $c_v\in \mathbb Z$ for all $v\in V(G)$, and  the addition in this group is naturally defined as
\[
  \sum_{v\in V(G)}{c_v[v]}+\sum_{v\in V(G)}{c'_v[v]}=\sum_{v\in V(G)}{(c_v+c'_v)[v]}.
\]
Similarly, let $\edge(\rel G)$ denote the free Abelian group with generators $[u,v], (u,v)\in E(G)$, where  we additionally postulate that $[u,v]=-[v,u]$ for every edge.
The elements of $\vertex(\rel G)$ are called \emph{vertex chains} and the elements of $\edge(\rel G)$ \emph{edge chains} in $\rel G$.
\end{definition}

Note that any multiset $W$ of oriented edges in $\rel G$ gives rise to the edge chain $\sum_{[u,v]\in W} [u,v]$, where each oriented edge appears in the sum with the corresponding multiplicity. With a~slight abuse of notation, we will denote this edge chain also by $W$.
For example, if $W$ is a walk that uses some edge $(u,v)$ the same number of times in each direction, then the corresponding coefficient in the edge chain of $W$ will be 0. 

Note also that one can consider both $\vertex(\rel G)$ and $\edge(\rel G)$ not only as Abelian groups, but also as $\mathbb Z$-modules. That is, for any integer $c$ and any vertex chain or edge chain $W$, one can consider the chain $c\cdot W$ defined by multiplying all coefficients in $W$ by $c$.

For any two graphs $\rel H$ and $\rel G$, any homomorphism $f\colon \rel H\to \rel G$ naturally gives rise to group homomorphisms \( f_\vt\colon \vertex(\rel H) \to \vertex(\rel G) \) and $f_\e\colon \edge(\rel H) \to \edge(\rel G)$ defined by
\[
  \textstyle f_\vt( \sum_i {c_i [v_i]} ) = \sum_i c_i [f(v_i)]
,\]
and
\[
  \textstyle f_\e( \sum_i c_i [u_i,v_i] ) = \sum_i c_i [f(u_i),f(v_i)]
.\]
Since $f$ is a~graph homomorphism, $[f(u_i),f(v_i)]$ is always an (orientation of an) edge in $\rel G$.

\begin{definition}
For a~graph $\rel G$, we define a~map $\bd\colon \edge(\rel G) \to \vertex(\rel G)$ as the group homomorphism such that
\(
  [u,v] \mapsto [v] - [u].
\)
for every $[u,v]\in \edge(\rel G)$
\end{definition}

Note that the above condition uniquely defines $\bd$.

The map $\bd$ computes the `boundary' of an edge chain. For example, the boundary of an edge chain corresponding to a~walk from $u$ to $v$ in $\rel G$ is $[v] - [u]$, and more  generally, the boundary $\delta W$ of an edge chain $W$ counts for each vertex $v$ the difference between how many times edges in $W$ arrive to $v$ and how many times they leave.

We will also use the following observation which is a~generalization of the fact that mapping a~walk from $u$ to $v$ by a~homomorphism $f$ results in a~walk from $f(u)$ to $f(v)$.

\begin{lemma}
\label{lem:boundary}
  For each graph homomorphism $f\colon \rel H \to \rel G$ and each $P\in \edge(\rel H)$, we have \( f_\vt (\bd P) = \bd f_\e(P) \).
\end{lemma}

\begin{proof}
  Since all the involved maps are group homomorphisms, it is enough to check the required equality on the generators of $\edge(\rel H)$. Pick $[u,v]$ an oriented edge of $\rel H$, then
  \[
    f_\vt (\bd [u,v]) = f_\vt( [v] - [u] ) = [f(v)] - [f(u)] = \bd [f(u),f(v)] = \bd f_\e([u,v]) 
  \]
  as required.
\end{proof}
\section{The Main Result}
  \label{sec:main}

Our main result is as follows.

\begin{theorem} \label{thm:main}
  Let $\rel H$ be a 3-colourable non-bipartite graph. Then $\PCSP(\rel H,\rel K_3)$ is \NP-hard.
\end{theorem}

As we explained in the introduction, it is enough to prove this theorem for the case $\rel H=\rel C_k$, $k\ge 3$ odd. We do this by using Theorem~\ref{thm:bounded-arity} for $\Pol(\rel C_k,\rel K_3)$ and the minion $\clo M$ defined as follows.

\begin{definition}
  Let $N$ be an odd number, we define a~minion $\clo Z_{\leq N}$ to be the set of all functions $f\colon \mathbb Z^n \to \mathbb Z$ such that
  \( 
  f(x_1,\dots,x_n) = c_1x_1 +\dots + c_nx_n
  \)
  for some $c_1$, \dots, $c_n\in \mathbb Z$ with $\sum_{i=1}^n \left|c_i\right| \leq N$ and $\sum_{i=1}^n c_i$ odd.
\end{definition}

Alternatively, the set $\clo Z_{\leq N}$ can be described as the set of all minors of the functions of the form
\(
  (x_1,\dots,x_N) \mapsto \pm x_1 \pm \dots \pm x_N
\).
It is clear that, for any fixed odd $N$, $\clo Z_{\leq N}$ is a~minion that has bounded essential arity and contains no constant function.

\begin{theorem} \label{thm:key}
  Let $k\geq 3$ be odd and let $N$ be the largest odd number such that
  $N\le k/3$. Then there is a~minion homomorphism from $\Pol(\rel C_k,\rel K_3)$ to $\clo Z_{\leq N}$.
\end{theorem}

\begin{remark}
If $k$ is the size of an odd cycle in $\rel H$, there also exists a~minion homomorphism $\xi\colon \Pol(\rel H,\rel K_3)\to \Pol(\rel C_k,\rel K_3)$, which can be composed with the minion homomorphism from Theorem~\ref{thm:key} to give a minion homomorphism from $\Pol(\rel H,\rel K_3)$ to $\clo Z_{\leq N}$. Given a graph homomorphism $h\colon \rel C_k \to \rel H$, we can define a~map $\xi\colon \Pol(\rel H,\rel K_3) \to \Pol(\rel C_k,\rel K_3)$ by
\[
  \xi(f) (x_1,\dots,x_n) = f( h(x_1),\dots,h(x_n) )
.\]
It is easy to show that this map preserves minors and is therefore a~minion homomorphism. 

The bound on $N$ given in the above theorem is tight. More precisely, one can show that there is also a~minion homomorphism in the opposite direction, i.e., from $\clo Z_{\leq N}$ to $\Pol(\rel C_k,\rel K_3)$ (see the appendix). It is not hard to check that this in particular implies that \cite[Corollary 5.19]{BKO18} cannot be used to provide \NP-hardness of $\PCSP(\rel C_k,\rel K_3)$ for any $k\geq 9$.
\end{remark}

As mentioned above, Theorem~\ref{thm:main} follows immediately from Theorem~\ref{thm:bounded-arity} and Theorem~\ref{thm:key}.
\section{A Topological Detour}
  \label{sec:topology}

The proof presented in Section~\ref{sec:proof} is heavily influenced by several topological observations, and even though they are not formally needed, we present them here to provide some intuition. The only intention of this section is to give an intuition about the combinatorial statements in the Section~\ref{sec:proof}, therefore we will omit any formal proofs or statements. We believe that an interested reader with an access to a~book about algebraic topology (e.g.\ \cite{Hat01}) will be able to check correctness of our statements. Throughout this section, we add a~few remarks intended for readers skilled with algebraic topology.

The analogy between our discrete setting and topology is based on the observation that both $\rel C_k$ for $k\geq 3$ and $\rel C_3$ look from the topological perspective like the~circle $S^1 = \{ (x,y) \in \mathbb R^2 \mid x^2 + y^2 = 1 \}$. Any continuous mapping $f\colon S^1 \to S^1$ is assigned a~topological invariant called \emph{degree} of $f$, and denoted by $\deg f$. Intuitively, this number counts `how many times $f$ loops around the circle'. A~positive degree means it loops around counter-clockwise, a~negative one means it loops around clockwise. A~similar invariant can be used for graph homomorphisms between two cycles (see~Definition~\ref{def:unary-degree}). The essence of our argument is to generalize this degree to polymorphisms, i.e., mappings that have multiple values on the input.

\begin{remark} In algebraic topology, the degree is formally defined through the \emph{fundamental group}. The fundamental group $\Pi_1(S^1)$ is isomorphic to the free cyclic group $\mathbb Z$, the generator of this group is the class of a~loop that loops around once counter-clockwise. Any continuous mapping $f\colon S^1 \to S^1$ induces a~group homomorphism between the fundamental groups, i.e., a~group homomorphism $\Pi_1(f)\colon \mathbb Z \to \mathbb Z$, and any such mapping is of the form $f(x) = cx$. This $c$ is then defined as the degree of $f$.
\end{remark}

\begin{figure}
  \includegraphics{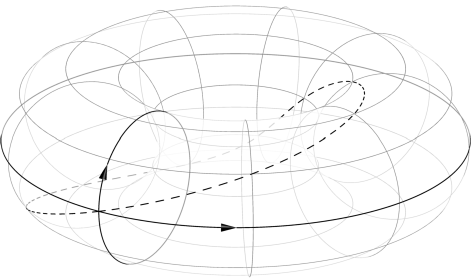}
  \caption{A torus with two representatives of coordinate loops. The dashed line represents points with coordinates $(x,x)$ for some~$x$.}
    \label{fig:three-loops}
\end{figure}

Let us borrow the term `polymorphism' to use for continuous mappings from a~power of a~topological space to another, i.e., a~polymorphism of our circle $S^1$ is a~continuous map from $n$-th power of a~circle to $S^1$ (with the product topology). The $n$-th power of $S^1$ is an \emph{$n$-torus}, usually denoted by $T^n$. The second power is the usual torus $T^2$ (surface of a~doughnut) depicted on Figure~\ref{fig:three-loops}. That is for $n$-ary polymorphisms, we are interested in continuous maps $f\colon T^n \to S^1$.

Such a~mapping $f\colon T^n \to S^1$ is assigned $n$ different degrees $\deg_1 f$, \dots, $\deg_n f$ each corresponding to one of the coordinates of $f$. A~\emph{degree of $f$ at a~coordinate $i$} is obtained by fixing all other coordinates to a~point, and following the $i$-th coordinate around $S^1$ and counting how many times one loops around the circle in the image. For example, for $n=2$, each of the two degrees are obtained by following one of the two loops depicted in Figure~\ref{fig:three-loops}. A~necessary observation is that degree assigned this way does not depend on the choice of values to which other coordinates are fixed. This is due to a~simple fact that any two such choices of loops can be connected by a~continuous transformation, continuously changing one loop into the other (such a~continuous transformation is usually called a~\emph{homotopy}), this implies that the degree has to change continuously as well. But the degree can only attain discrete values, and therefore it has to remain constant.

This assigns a~quantity $\deg_i f$, which is always an integer, to each of the coordinates of $f$. Intuitively, we can say that the higher the absolute value of this degree is, the more important the corresponding variable is. In particular, inessential variables have degree $0$. This is in essence how we identify which variables of a function are important, and which should become inessential after applying a minion homomorphism.

\begin{remark} Using the fundamental groups in the $n$-ary case can also bring a~little more insight. In particular, as it is well-known that $\Pi_1(T^n)$ is isomorphic to the $n$-generated free Abelian group. The loops that we described in the above paragraphs correspond to the $n$ different generators. And similarly, as in the unary case, any continuous mapping $f\colon T^n \to S^1$ induces a~group homomorphism between the fundamental groups, i.e., a~group homomorphism $\Pi_1(f)\colon \mathbb Z^n \to \mathbb Z$. Any such map is of the form
\[
  (x_1,\dots,x_n) \mapsto c_1x_1 + \dots + c_nx_n
,\]
and each of these coefficients $c_i$ correspond to the degree $\deg_i f$.
\end{remark}

To bound the number of `interesting' coordinates, we need use the discrete structure of the graph. One easy observation is that a~degree of a~graph homomorphism from $\rel C_k$ to $\rel C_3$ cannot be arbitrarily large: we can walk around the cycle $\rel C_3$ at most $k/3$ times in $k$ steps. We need to bring this bound on a~single degree of a~unary map to bound the number of coordinates with non-zero degree. This is done by proving that if $f$ is $n$-ary, and $g$ is defined from $f$ by identifying all variables, i.e.,
\(
  g(x) = f(x,\dots,x)
\),
then $\deg g = \deg_1 f + \dots + \deg_n f$. This is not so easy to see, let us sketch the proof for $n=2$. Let $f\colon T^2 \to S^1$, then $g$ is defined as the restriction of $f$ to the diagonal, i.e., points with coordinates $(x,x)$, see the dashed line in Figure~\ref{fig:three-loops}.

We want to connect the degree of this restriction with the degrees of the two restrictions of $f$ to the loops that define $\deg_1 f$ and $\deg_2 f$. This is again done by observing that walking along the two loops one after another can be continuously transformed to walking along the diagonal. This can be done by continuously shifting the walk along the lines shown in Figure~\ref{fig:threeloops3}. A~similar argumentation works for higher dimensions as well.
The last small technical obstacle is what to do with negative degrees as they could cancel out with positive ones. This is only a minor problem since we can simply reverse the corresponding coordinates to obtain a~mapping that has only positive degrees that are identical to the original ones, up to a~sign.

\begin{figure}
  \includegraphics{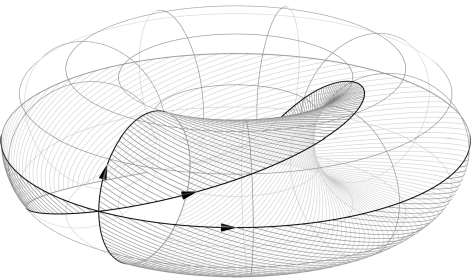}
  \caption{A homotopy.}
    \label{fig:threeloops3}
\end{figure}

\begin{remark} The above argumentation is an~instance of a~more general statement that says that the mapping $\Pi_1\colon \clo S^1 \to \clo Z$ (here $\clo S^1$ denotes the minion of all continuous maps from $T^n$ to $S^1$ and $\clo Z$ the minion of all group homomorphisms from $\mathbb Z^n$ to $\mathbb Z$) is a~minion homomorphism.
In other words, if $g$ is defined from $f$ using $\pi\colon \{1,\dots,n\} \to \{1,\dots,m\}$ by
\(
  g(x_1,\dots,x_m) = f(x_{\pi(1)},\dots,x_{\pi(n)})
\),
then the same identity holds for $\Pi_1(g)$ and $\Pi_1(f)$, i.e.,
\[
  \deg_1 g \cdot x_1 \dots + \deg_m g \cdot x_m =
  \deg_1 f \cdot x_{\pi(1)} \dots + \deg_n f \cdot x_{\pi(n)}
.\]
The above is equivalent to the statement that for all $i\in\{1,\dots,m\}$ we have
\[
  \textstyle \deg_i g  = \sum_{j\in \pi^{-1}(i)} \deg_j f
.\]
In Section~\ref{sec:proof}, we prove that the degrees we define for graph polymorphisms also have this property.
\end{remark}

\begin{figure}
  \includegraphics{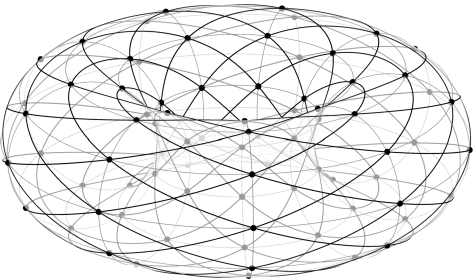}
  \caption{The graph $\rel C_9^2$ on a~torus.}
    \label{fig:c9-squared}
\end{figure}

In our attempt to bring these topological considerations to proper statements about polymorphisms from $\rel C_k$ to $\rel C_3$, there are a~few points where the analogy does not work nicely. We already mentioned one, that the~degree of a~graph homomorphism is bounded, but the~degree of a~continuous map $S^1 \to S^1$ is not. This is due to the fact that, unlike topological spaces which are sometimes described as `being made of rubber', i.e., they can be infinitely stretched and folded, graphs are `made of sticks', i.e., they can be folded but not stretched. This property works to our advantage. The second issue is that the second power of $\rel C_k$ is not exactly topologically equivalent to a~torus, rather it forms a~certain mesh that can be drawn on a~torus in some way (see Figure~\ref{fig:c9-squared}). This can be avoided by using a~less na\"ive assignment of a topological space to a~graph and a~more robust theory \cite{Wal83-equivariant} which is in line with Lov\'asz's approach \cite{Lov78}. This approach is described in detail in a~subsequent paper of Wrochna and Živný \cite[Appendix A]{WZ19}. In this paper, we stick to the na\"ive approach and present an ad-hoc (`discrete continuity') argument using an alternative definition of a~degree that resembles the topological one.
\section{Proof of Theorem \ref{thm:key}}
  \label{sec:proof}

We prove the theorem by analysing the polymorphisms from $\rel C_k$ to $\rel C_3 (=\rel K_3)$, where $k\ge 3$ is  an odd number (which we assume to be fixed for the rest of this section.

\subsection{Degree of a homomorphism}

Recall that, for $m\geq 3$, we define the~graph $\rel C_m$ to be the $m$-cycle with vertices $0,1,\dots,m-1$. Here vertices are connected by an edge if they differ by exactly one modulo $m$.

We fix an orientation of any $\rel C_m$ in the increasing order modulo $m$, and denote by $O_m$ the edge chain  $[0,1]+[1,2]+ \dots +[m-1,0]$ in $\edge(\rel C_m)$.

The \emph{degree} of a~homomorphism $f\colon \rel C_m \to \rel C_l$ is intuitively defined as the (possibly non-positive) number of times the image of $\rel C_m$ under $f$ walks around $\rel C_l$ in a fixed direction (say, counter-clockwise). The formal definition is based on the following observation.

\begin{lemma} \label{lem:degree-def}
Let $m,l \geq 3$, and let $f\colon \rel C_m \to \rel C_l$ be a~homomorphism. Then there is an integer $d$ such that
\(
  f_\e(O_m) = d \cdot  O_l
\).
\end{lemma}

\begin{proof} Clearly, we have $\bd O_m = 0$. Lemma~\ref{lem:boundary} then implies that $\bd(f_\e(O_m)) = 0$. We claim that the only edge chains $W$ in $\edge(\rel C_l)$ such that $\bd W = 0$ are chains of the form $d\cdot O_l$, so $f_\e(O_m)$ is of this form. Indeed, observe that if
\(
  W = d_0 [0,1] + \cdots + d_{l-1} [l-1,0]
\),
then
\begin{multline*}
  \bd W = d_0  ([1] - [0]) + d_1  ([2] - [1]) + \dots + d_{l-1}  ([0] - [l-1]) \\
    = (d_{l-1} - d_0)  [0] + (d_0 - d_1)  [1] + \dots + (d_{l-2} - d_{l-1}) [l-1]
.\end{multline*}
If $\bd W = 0$, all coefficients in the above sum are $0$, and therefore $d_0 = d_1 = \dots = d_{l-1}$ concluding that $W = d\cdot O_l$ for $d = d_0$.
\end{proof}

\begin{definition} \label{def:unary-degree}
Let $m,l\geq 3$. The~\emph{degree} of a~homomorphism $f\colon \rel C_m \to \rel C_l$, denoted by $\deg f$, is defined as the~integer $d$ from the above lemma, i.e., the~number $\deg f$ such that
\[
  f_\e(O_m) = \deg f \cdot O_l.
\]
\end{definition}

\begin{lemma} \label{lem:unary-degree}
  Let $m,l\geq 3$, assume that $l$ is odd, and let $f\colon \rel C_m \to \rel C_l$ be a~homomorphism. Then
  \begin{enumerate}
    \item $\mathopen|\deg f\mathclose| \leq m/l$,
    \item the parity of $\deg f$ is the same as the parity of $m$, and
    \item if $m = 4$ then $\deg f = 0$.
  \end{enumerate}
\end{lemma}

\begin{proof}
    (1)\quad We have that
    \[
      \deg f\cdot O_l=f_\e(O_m) = [f(0),f(1)] + [f(1),f(2)] + \dots +[f(m-1),f(0)].
    \]
    It is clear that, for each $[i,i+1]$ in $O_l$, the last expression above contains at least $\mathopen|\deg f\mathclose|$ terms that are either $[i,i+1]$ or $[i+1,i]$. It follows that $\mathopen|\deg f\mathclose|\le m/l$.

    (2)\quad This follows by similar considerations as above. For each $[i,i+1]$ in $O_l$, the parity of the number of terms in the above sum that are either $[i,i+1]$ or $[i+1,i]$ is the same as the parity of $\deg f$. Since $l$ is odd, the result follows.
    
    (3)\quad From (1) and (2), we know that the degree of any homomorphism $f\colon \rel C_4 \to \rel C_l$ is an even integer with absolute value at most $4/l$. For $l > 2$, there is only one such number, namely $0$.
\end{proof}

Note that as a~direct consequence of item (2) in the above lemma, we get that for $m,l$ odd, any  homomorphism from $\rel C_m$ to $\rel C_l$ has a~non-zero degree.

\subsection{Degrees of a polymorphism}

We generalise the notion of a~degree of a~homomorphism to polymorphisms between odd cycles. More precisely, for a~polymorphism $f\colon \rel C_k^n \to \rel C_3$ and coordinate $i\in \{1,\dots,n\}$, we define a~quantity that we will call `a~degree of $f$ at coordinate $i$'. Since this quantity will be used to define a~minion homomorphism, the main requirement here will be that the degree behaves nicely with respect to minors. Formally, we will need that if $g$ is a~minor of $f$ defined by
\[
  g(x_1,\dots,x_m) = f(x_{\pi(1)},\dots,x_{\pi(n)})
,\]
then
\[
  \textstyle \deg_i g = \sum_{j \in \pi^{-1}(i)} \deg_j f
.\]
This property is equivalent to saying that the mapping that maps $f$ to the function on $\mathbb Z$ defined by
\( (x_1,\dots,x_n) \mapsto \deg_1 f\cdot x_1 + \dots + \deg_n f \cdot x_n \)
is minor-preserving.

Intuitively, a~degree of a~unary function counts how many times one loops around the cycle $\rel C_3$ if one follows the values of the function. We would like to bring this intuition to the $n$-ary case, so that the degree of $f$ at some coordinate would mean `number of times one loops around the circle if one follows edges going in the given direction at this coordinate'. 

We will formalise this intuition and prove that the degree at a~coordinate can be defined in two equivalent ways, one global and the other local. In what follows we fix $l=3$, but all proofs work for any odd $3\le l\le k$.

We denote by $O_{k,i}^n$ the set of all oriented edges of $\rel C_k^n$ whose $i$-th coordinate is oriented as in $O_k$, i.e., 
\begin{multline*}
  O_{k,i}^n=\{ [(a_1,\dots,a_n),(b_1,\dots,b_n)] \mid \\
    (a_j,b_j) \in E(C_k) \text{ for all $j\in \{1,\dots,n\}$ and } [a_i,b_i] \in O_k \}
.\end{multline*}

We will also view $O_{k,i}^n$ as an edge chain in 
$\edge(\rel C_k^n)$.

\begin{definition} \label{def:n-ary-degree}
Let $f\colon \rel C_k^n \to \rel C_3$ be a~polymorphism. We define the \emph{degree of $f$ at coordinate $i$} as the integer $\deg_i f$ such that
\[
  f_\e(O_{k,i}^n) = (2k)^{n-1} \deg_i f \cdot O_3
.\]
\end{definition}

Note that $\mathopen| O_{k,i}^n \mathclose | = 2^{n-1}k^n$, and therefore the above definition agrees with the intuitive meaning. Also if $n=1$, then $\deg_1 f$ coincides with $\deg f$ since $(2k)^{1-1} = 1$ and $O_{k,1}^1 = O_k$.
For a~general $n$, it is not even clear that such a~number $\deg_i f$ always exists. It is easy to show that there is an integer $d'$ such that
\(
  f_\e(O_{k,i}^n) = d' \cdot O_3
\)
since $\bd(O_{k,i}^n) = 0$. However, there is no obvious reason that this number is a~multiple of $(2k)^{n-1}$. Let us show that this is the case. We need a~technical definition first.

\begin{definition}
For an unoriented edge $e = (\bar{u},\bar{v})$ of $\rel C_k^{n-1}$, we define
\begin{multline*}
  e \times_i O_k = \{ [(u'_1,\dots,u'_n),(v'_1,\dots,v'_n)] \mid \\
    \{ (\dots,u'_{i-1},u'_{i+1},\dots), (\dots,v'_{i-1},v'_{i+1},\dots) \}
      = \{\bar{u},\bar{v}\}, \\
    [u'_i,v'_i] \in O_k
    \}.
\end{multline*}
\end{definition}

Note that $\mathopen| e \times_i O_{k,i}^n \mathclose| = 2k$ since for each $[u_i,v_i] \in O_k$ there are two edges in $e \times_i O_{k,i}^n$ (one for each orientation of $e$) whose $i$-th coordinate agree with $[u_i,v_i]$. We also note that $O_{k,i}^n = \bigcup_{e} e \times_i O_k$ where the union runs through all unoriented edges of $\rel C_k^{n-1}$.
Again, we can view $e \times_i O_k$ as an edge chain in  $\edge(\rel C_k^n)$.

\begin{figure}
  \begin{centering}
  \begin{tikzpicture}[label distance=.25cm, scale = 1.5]
    \draw [thick,shorten > = .5cm, shorten < = .5cm] (-3,1) -- (-1,-1) -- (1,1) -- (3,-1);
    \draw [thick,shorten > = .5cm, shorten < = .5cm] (-3,-1) -- (-1,1) -- (1,-1) -- (3,1);
    \draw [-latex,shift=(170:.5cm)] (1,0) arc (170:-170:.5cm);
    \node at (1,0) {$S_{j+2}$};
    \draw [-latex,shift=(170:.5cm)] (-1,0) arc (170:-170:.5cm);
    \node at (-1,0) {$S_j$};

    \draw [-latex,rounded corners = .4cm] (-2.4,.7) -- ++(.4,-.4) -- ++(1,1) -- ++(1,-1)  -- ++(1,1) -- ++(1,-1) -- ++(.4,.4);
    \node at (2,1) {$e\times O_k$};
    \draw [-latex,rounded corners = .4cm] (-2.4,-.7) -- ++(.4,.4) -- ++(1,-1) -- ++(1,1)  -- ++(1,-1) -- ++(1,1) -- ++(.4,-.4);
    \node at (2,-1) {$e'\times O_k$};

    \node [circle,fill,inner sep=1.5,label={[scale=.9]left:{$(j,u_2,\dots,u_n)$}}] at (-2,0) {};
    \node [circle,fill,inner sep=1.5] at (0,0) {};
    \node [circle,fill,inner sep=1.5,label={[scale=.9]right:{$(j+4,u_2,\dots,u_n)$}}] at (2,0) {};
    \node [circle,fill,inner sep=1.5,label={[scale=.9]above:{$(j+1,v_2,\dots,v_n)$}}] at (-1,1) {};
    \node [circle,fill,inner sep=1.5,label={[scale=.9]above:{$(j+3,v_2,\dots,v_n)$}}] at (1,1) {};
    \node [circle,fill,inner sep=1.5,label={[scale=.9]below:{$(j+1,w_2,\dots,w_n)$}}] at (-1,-1) {};
    \node [circle,fill,inner sep=1.5,label={[scale=.9]below:{$(j+3,w_2,\dots,w_n)$}}] at (1,-1) {};
  \end{tikzpicture}
  \caption{Proof of Lemma \ref{lem:poly-degree}(2).}
    \label{fig:poly-degree}
  \end{centering}
\end{figure}
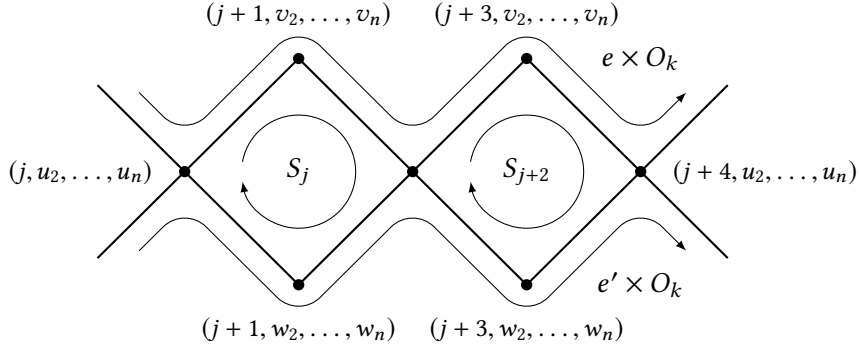

\begin{lemma}
  \label{lem:poly-degree}
  Let $n\geq 2$, $f\colon \rel C_k^n \to \rel C_3$ be a~polymorphism, and let $i\in \{1,\dots,n\}$. Then
  \begin{enumerate}
    \item for each edge $e$ of $\rel C_k^{n-1}$, there is an~integer $d$ such that $f_\e( e\times_i O_k ) = 2d \cdot O_3$;
    \item the above $d$ does not depend on the choice of $e$;
    \item $d = \deg_i f$.
  \end{enumerate}
\end{lemma}

\begin{proof}
  Without loss of generality, assume that $i=1$, and to simplify the notation, we will write $\times$ instead of $\times_1$.

  (1)\quad Observe that $e \times O_k$ is an oriented $2k$-cycle in $\rel C_k^n$, and consider $g_e\colon \rel C_{2k} \to \rel C_3$ to be the restriction of $f$ to this $2k$-cycle. Then $\deg g_e$ is even from Lemma~\ref{lem:unary-degree}(2), and therefore it is equal to $2d$ for some $d$.

  (2)\quad We first prove the claim for two incident edges $e$ and $e'$. Let
  \[
    e = ((u_2,\dots,u_n), (v_2,\dots,v_n))
    \text{ and }
    e' = ((u_2,\dots,u_n), (w_2,\dots,w_n))
  .\]
  We want to prove that $f_\e( e \times O_k ) = f_\e( e' \times O_k )$ which is equivalent to $f_\e( e\times O_k - e' \times O_k ) = 0$ since $f_\e$ is a~group homomorphism. Note that $-e' \times O_k$ is obtained from $e'\times O_k$ by reversing edges.
  Our goal is then decompose these two oriented cycles into several 4-cycles and then apply Lemma~\ref{lem:unary-degree}(3). The four cycles are defined on vertices
  \[
    (j,u_2,\dots,u_n),
    (j+1,v_2,\dots,v_n),
    (j+2,u_2,\dots,u_n),
    (j+1,w_2,\dots,w_n)
  \]
  where the addition in the first coordinate is considered modulo $k$. We denote by $S_j$ the sum of oriented edges of the above 4-cycle, with the orientation following the order above. Observe that indeed (see Figure~\ref{fig:poly-degree})
  \[
    \sum_{j<k} S_j = e\times_i O_k - e' \times_i O_k
  ,\]
  and therefore
  \[
    f_\e( e \times_i O_k - e' \times_i O_k ) = f_\e( \sum_{j<k} S_j ) = \sum_{j<k} f_\e( S_j ) = 0
  \]
  where the last equality follows from Lemma~\ref{lem:unary-degree}(3). This implies that $f_\e(e \times_i O_k) = f_\e(e' \times_i O_k)$, as required. The general case is then obtained by transitivity,
  since one can move from any edge of $\rel C_k^{n-1}$ to any other edge by following a sequence of incident edges.

  (3)\quad We have
  \[
    f_\e ( O_{k,i}^n) =
      \sum_{e\in E(C_k^{n-1})} f_\e (e\times_i O_k) =
      \sum_{e\in E(C_k^{n-1})} 2d\cdot O_3 =
      (2k)^{n-1} d \cdot O_3
  \]
  since $\lvert E(C_k^{n-1})\rvert = (2k)^{n-1}/2$.
\end{proof}

\subsection{Minor preservation}
  \label{sec:minor-preservation}

Let $\clo Z$ denote the minion of all linear maps over $\mathbb Z$, i.e., of the functions of the form $\sum{c_ix_i}$ where all $c_i\in \mathbb Z$.
We define a~mapping $\delta\colon \Pol(\rel C_k,\rel C_3) \to \clo Z$  by
\[
  \delta(f)\colon (x_1,\dots,x_n) \mapsto \deg_1 f\cdot x_1 + \dots + \deg_n f\cdot x_n.
\]
In this subsection we prove that $\delta$ is minor-preserving, and therefore a~minion homomorphism. In the following one we show that the image of $\delta$ contains only functions of bounded essential arity (but no constant function).

\begin{lemma} \label{lem:minor-preservation}
  The map $\delta$ is a~minion homomorphism. 
\end{lemma}

It is clear that $\delta$ preserves the arity, so we need to show it also preserves the operation of taking minors. We decompose this operation into a~few steps: permuting variables, introducing new dummy variables, and identifying two variables. It is not hard to observe that $\delta$ preserves the operation of permuting variables (this corresponds to the case when $\pi$ in Definition~\ref{def:minion-homomorphism} is a bijection). We deal with the case of identifying two variables in Lemma~\ref{lem:sum-of-degrees}, and then consider the addition of dummy variables in Lemma~\ref{lem:dummy}.

\begin{lemma} \label{lem:sum-of-degrees}
  Let $n\geq 2$. If $f\in \Pol(\rel C_k,\rel C_3)$ is $n$-ary and $g$ is obtained from $f$ by identifying the first two variables, i.e.,
  \[
    g(y,x_3,\dots,x_n) = f(y,y,x_3,\dots,x_n)
  \]
  then $\deg_1 g = \deg_1 f + \deg_2 f$.
\end{lemma}

Before, we get to the proof, we need some technical definitions and a~simple technical lemma.
Similarly to $O_{k,i}^n$, we denote by $O_{k,\{1,2\}}^n$ the set of all oriented edges of $\rel C_k^n$ whose first and second coordinates are oriented as in $O_k$, i.e.,
\[
  O_{k,\{1,2\}}^n = O_{k,1}^n \cap O_{k,2}^n
.\]
Note that, unlike $O_{k,i}^n$, $O_{k,\{1,2\}}^n$ does not contain all edges of $\rel C_k^n$ in some orientation, e.g.\ neither the edge $[(0,1),(1,0)]$ nor $[(1,0),(0,1)]$ is contained in $O_{k,\{1,2\}}^2$. Also observe that when considering this set as an edge chain, we have
\[
  O_{k,\{1,2\}}^n = 1/2\cdot (O_{k,1}^n + O_{k,2}^n)
\]
which follows, since in the sum on the right-hand side the edges that disagree in the orientation in the first two coordinates cancel out, and those which agree count twice.

We define the joint degree of $f$ at coordinates $1$ and $2$, which intuitively expresses `the average number of times one loops around $O_3$ when following $k$ edges that increase in both coordinates $1$ and $2$'. For the formal definition below, note that $|O_{k,\{1,2\}}^n| = 2^{n-2}k^n$.

\begin{definition} Let $n\geq 2$, and let $f\colon \rel C_k^n \to \rel C_3$ be a~polymorphism. We define a~\emph{joint degree} of $f$ at coordinates 1 and 2 as the integer $\deg_{1,2} f$ such that
\[
  f_\e(O_{k,\{1,2\}}^n) = 2^{n-2}k^{n-1}\deg_{1,2} f \cdot O_3
.\]
\end{definition}

As in the case of degrees of polymorphisms, it is not obvious that such a~number exists, but we prove this in the following lemma.

\begin{lemma} \label{lem:joined-degree}
  For each $n\geq 2$ and a~polymorphism $f\colon \rel C_k^n \to \rel C_3$, we have
  \[
    \deg_{1,2} f = \deg_1 f + \deg_2 f.
  \]
\end{lemma}

\begin{proof}
  Since $O_{k,\{1,2\}}^n = 1/2\cdot (O_{k,1}^n + O_{k,2}^n)$, we have
  \[
    f_\e(O_{k,\{1,2\}}^n) = 1/2\cdot f_\e(O_{k,1}^n + O_{k,2}^n)
  ,\]
  and consequently,
  \[
    2^{n-2}k^{n-1} \deg_{1,2} f = 1/2\cdot ((2k)^{n-1}\deg_1 f + (2k)^{n-1}\deg_2 f)
  .\]
  Cancelling $2^{n-2}k^{n-1}$ on both sides gives the claim.
\end{proof}

  \begin{figure}
  \begin{centering}
    \begin{tikzpicture}[label distance=.25cm, scale = 1.5]
      \draw [thick,shorten > = .5cm, shorten < = .5cm] (-3,1) -- (1,-3);
      \draw [thick,shorten > = .5cm, shorten < = .5cm] (-2,2) -- (2,-2);
      \draw [thick] (-1,1) -- (-2,0);
      \draw [thick] (0,0) -- (-1,-1);
      \draw [thick] (1,-1) -- (0,-2);

      \draw [-latex,shift=(170:.5cm)] (-1,0) arc (170:-170:.5cm);
      \node at (-1,0) {$S_x$};
      \draw [-latex,shift=(170:.5cm)] (0,-1) arc (170:-170:.5cm);
      \node at (0,-1) {$S_{x+1}$};

      \draw [-latex] (-1.35,1.65) -- ++(3,-3);
      \node at (2,-1.5) {$O_{k,i}'$};
      \draw [-latex] (-2.65,.35) -- ++(3,-3);
      \node at (.5,-3) {$O_{k,i+2}'$};

      \node [circle,fill,inner sep=1.5,label={above right:$(x  ,x+i  )$}] at (-1,1) {};
      \node [circle,fill,inner sep=1.5,label={above right:$(x+1,x+i+1)$}] at (0,0) {};
      \node [circle,fill,inner sep=1.5,label={above right:$(x+2,x+i+2)$}] at (1,-1) {};
      \node [circle,fill,inner sep=1.5,label={below left:$(x-1,x+i+1)$}] at (-2,0) {};
      \node [circle,fill,inner sep=1.5,label={below left:$(x  ,x+i+2)$}] at (-1,-1) {};
      \node [circle,fill,inner sep=1.5,label={below left:$(x+1,x+i+3)$}] at (0,-2) {};
    \end{tikzpicture}
    \caption{Proof of Case 1 of Lemma~\ref{lem:sum-of-degrees}.}
      \label{fig:sum-of-degrees-1}
  \end{centering}
  \end{figure}
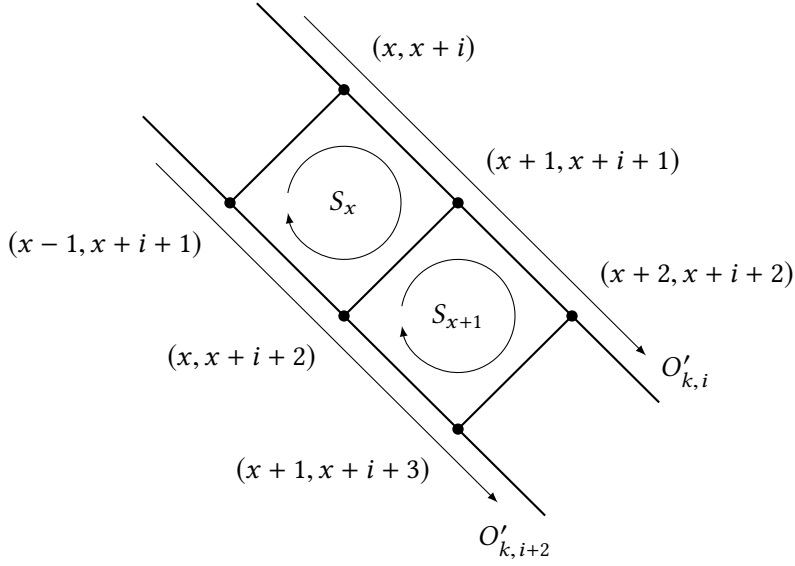

\begin{proof}[Proof of Lemma~\ref{lem:sum-of-degrees}]
  By the previous lemma, it is enough to prove that $\deg_{1,2} f = \deg g$. This is done in a~similar way to proving that the degree of $f$ at a coordinate is both a~local and a~global property of $f$ (Lemma~\ref{lem:poly-degree}). We prove this statement separately for two cases: (1) $f$ is binary and $g$ is unary; and (2) $f$ has arity at least $3$ and $g$ has arity at least $2$. The two cases are very similar. We present them separately to ease some technical difficulties of the proof. 

  Case 1: $f$ is binary. The assumption says that $g(x) = f(x,x)$, and we aim to prove that $\deg g = \deg_{1,2} f$. Note that
  \[
    g_\e( O_k ) = f_\e( O'_{k,0} )
  \]
  where $O'_{k,0}$ is the set of oriented edges of the $k$-cycle $(0,0),(1,1), \dots, (k-1,k-1)$ in $\rel C_k^2$ (in the increasing orientation). Note that $O'_{k,0} \subseteq O_{k,\{1,2\}}^2$. More generally, for $i \in \{1,\dots,k-1\}$, we denote by $O'_{k,i}$ the set of oriented edges of the $k$-cycle $(0,i), (1,i+1), \dots, (k-1,i-1)$ where in the second coordinate, $0$ succeeds $k-1$. Thus, the set $O_{k,\{1,2\}}^2$ is a disjoint union of the cycles $O'_{k,i}$, $0\le i\le k-1$.
  Since each $O'_{k,i}$ is a~$k$-cycle, we know that there is $d_i$ such that
  \[
    f_\e ( O'_{k,i} ) = d_i \cdot O_3.
  \]
  In a~similar way as in Lemma~\ref{lem:poly-degree}, we prove that $d_i$ does not depend on $i$ and is actually equal to $\deg_{1,2} f$. Also note that $d_0 = \deg g$.

  First, we fix any $0\le i\le k-1$ and show that we have $d_i = d_{i+2}$ where the addition is modulo $k$. For each $x = 0,\dots,k-1$, let $S_x$ be the set of edges of the oriented 4-cycle
  \[ 
    (x,x+i),(x+1,x+i+1), (x,x+i+2), (x-1,x+i+1)
  \]
  where the addition is considered modulo $k$.
  Observe that
  \(
    \sum_{x=0}^{k-1} S_x = O'_{k,i} - O'_{k,i+2}
  \)
  (see Figure~\ref{fig:sum-of-degrees-1}). By Lemma~\ref{lem:unary-degree}(3), we have $f_\e(S_x) = 0$ for each $x$. So we get
  \[
    0 = f_\e( \sum_{x=0}^{k-1} S_x ) = f_\e( O'_{k,i} - O'_{k,i+2} ) = f_\e( O'_{k,i} ) - f_\e( O'_{k,i+1} )
  \]
  and therefore $f_\e( O'_{k,i} ) = f_\e ( O'_{k,i+2} )$ which implies that $d_i = d_{i+2}$. Since $k$ is odd, it follows that $d_i$ is the same for for all $i$. This implies that $d_i = \deg g$, and therefore
  \[
    f_\e( O_{k,\{1,2\}}^2 ) = \sum_{i=0}^{k-1} f_\e ( O'_{k,i} ) = \sum_{i=0}^{k-1} \deg g\cdot O_3 = k\deg g \cdot O_3
  .\]
  Consequently, $\deg_{1,2} f = \deg g$ which concludes the proof of the first case.

  Case 2: The polymorphism $f$ is of arity $n > 2$. The proof is similar to the first case: We decompose the set $O_{k,\{1,2\}}$ into $k$ disjoint sets each of which is a~copy of $O_{k-1,1}^{n-1}$ used in the definition of $\deg_1 g$. One of these copies will exactly correspond to the edges of $\rel C_k^n$ on vertices that have the first two coordinate identical. This is important since $g$ is essentially the restriction of $f$ to such vertices.

  The sets $P'_i$ are defined to consist of those oriented edges of $\rel C_k^n$ whose projection to the first two coordinates is in $O'_{k,i}$. In other words, we put
  \begin{multline*}
    P'_i = \{ [(x,x+i,u_3,\dots,u_n),(x+1,x+i+1,v_3,\dots,v_n)] \mid \\
      x = 0,\dots,k-1,\text{ and }
      (u_j,v_j) \in E(C_k) \text{ for all $j\geq 3$ }
    \}
  \end{multline*}
  (again, the addition in the first two coordinates is considered modulo $k$).
  It is easy to see that $\bd P'_i = 0$, $O_{k,\{1,2\}} = \Union_{i=0}^{k-1} P'_i$, and
  \(
    g_\e( O_{k,1}^{n-1} ) = f_\e( P'_0 )
  \).

  We prove $f_\e(P'_i) = f_\e(P'_{i+2})$.
  For each edge
  \[
    e = [(x,x+i,u_3,\dots,u_n),(x+1,x+i+1,v_3,\dots,v_n)] \in P_i'
  ,\]
  let $S_e$ denote the set of oriented edges of the 4-cycle
  \begin{multline*}
    (x,x+i,u_3,\dots,u_n),(x+1,x+i+1,v_3,\dots,v_n), \\
    (x,x+i+2,u_3,\dots,u_n), (x-1,x+i+1,v_3,\dots,v_n)
  .\end{multline*}
  Note that $e\in S_e$, and also
  \(
    [(x,x+i+2,u_3,\dots,u_n), (x-1,x+i+1,v_3,\dots,v_n) ] \in -P'_{i+2}
  \).
  We claim that
  \(
    \sum_{e\in P_i'} S_e = P_i' - P_{i+2}'
  \). In other words, all edges of the $S_e$'s not contained in
  $P_i' - P_{i+2}'$ cancel out. Indeed, every edge of the form
  \(
    d = [(x+1,x+i+1,v_3,\dots,v_n), (x,x+i+2,u_3,\dots,u_n)]
  \)
  that appears in $S_e$, has its reverse in $S_{e^+}$ for 
  \[
    e^+ = [(x+1,x+i+1,u_3,\dots,u_n), (x+2,x+i+2,v_3,\dots,v_n)]
  .\]
  Furthermore, the correspondence $e \leftrightarrow e^+$ is 1-to-1, and therefore we paired each $d\in S_e$ with a~unique $-d\in S_{e^+}$.

  We conclude the proof in a similar way as Case 1: From Lemma~\ref{lem:unary-degree}(3), we have $f_\e(S_e) = 0$ for each $e$, and therefore
  \[
    0 = f_\e( \sum_{e\in P'_i} S_e ) = f_\e( P_i' - P_{i+2}' )
  \]
  which shows that $f_\e(P_i') = f_\e(P_{i+2}')$. Consequently, we have $P_i' = P_j'$ for all $i,j$, and
  \[
    f_\e( O_{k,\{1,2\}}^n ) = \sum_{i<k}  f_e ( P_i' ) = k\cdot f_e (P_0') =
      k\cdot g_\e(O_{k,1}^{n-1}) = k(2k)^{n-2} \deg_1 g \cdot O_3
  .\]
  This implies that $\deg_{1,2} f = \deg_1 g$, as required.
\end{proof}

The following lemma says that $\delta$ preserves the operation of adding a~dummy variable.

\begin{lemma} \label{lem:dummy}
  Let $f\colon \rel C_k^n \to \rel C_3$ be a~polymorphism, and $i\in\{1,\dots,n\}$. If the coordinate $i$ in $f$ is dummy, then $\deg_i f = 0$.
\end{lemma}

\begin{proof}
  Loosely speaking, the degree of $f$ at the $i$-th coordinate is determined by the image of $O_{k,i}^n$ under $f_\e$. Since $f$ does not depend on the $i$-th coordinate, neither does (in the corresponding meaning) $f_\e$, and therefore, it cannot distinguish the orientation of the edges in $O_{k,i}^n$.

  Formally, we use the local definition of $\deg_i f$ and prove that, for an edge $e \in E(C_k^{n-1})$, we have
  \(
    f_\e( e\times_i O_k ) = 0
  \).
  Fix $e = ((u_2,\dots,u_n),(v_2,\dots,v_n))$ and assume without loss of generality that $i=1$. As mentioned in the proof of Lemma~\ref{lem:poly-degree}(1), $e\times_1 O_k$ is a~$2k$-cycle. 
  The values of $f$ on this cycle are all of the form $f(a,u_2,\dots,u_n)$ and $f(a,v_2,\dots,v_n)$ where $a\in V(C_k)$.  Since the first coordinate in $f$ is dummy, none of this values depends on $a$, so $f$ attains only two possible values on this $2k$-cycle.
  This implies that necessarily $f_\e(e \times_1 O_k) = 0\cdot O_3$, and consequently, by Lemma~\ref{lem:poly-degree}, $\deg_1 f = 0$.
\end{proof}

This concludes the proof of Lemma~\ref{lem:minor-preservation}.

\subsection{Bounding the essential arity}

To finish the analysis of polymorphisms from $\rel C_k$ to $\rel C_3$, we need to bound the essential arity of functions in the image of $\delta$ (defined in Section~\ref{sec:minor-preservation}) and show that none of these functions is a constant function. 
We now prove that the image of $\delta$ is contained in $\clo Z_{\leq N}$, where $N$ is the largest odd number such that $N\le k/3$.
Recall that, for an odd number $N$, the~minion $\clo Z_{\leq N}$ is defined to be the set of all functions $f\colon \mathbb Z^n \to \mathbb Z$ of the form
\( 
  f(x_1,\dots,x_n) = c_1x_1 +\dots + c_nx_n
\)
for some $c_1$, \dots, $c_n\in \mathbb Z$ with $\sum_{i=1}^n \mathopen|c_i\mathclose| \leq N$ and $\sum_{i=1}^n c_i$ odd.
It is easy to see that the number of non-zero coefficients in any function from $\clo Z_{\leq N}$ is between 1 and $N$, giving us the required result. 
We remark that our bound on $N$ is tight (for a proof of this fact, see Lemma~\ref{lem:app_b}), but we will not need that.

\begin{lemma}
\label{lem:bounding}
  Let $N$ be the largest odd number such that $N\le k/3$. Then we have $\delta(f)\in \clo Z_{\leq N}$ for all  $f\in \Pol(\rel C_k,\rel C_3)$.
\end{lemma}

\begin{proof}
We need to prove
(1) $\sum_{i=1}^n \mathopen|\deg_i f\mathclose| \leq N$, and
(2) $\sum_{i=1}^n \deg_i f$ is odd.
Consider the unary minor of $f$, i.e., the mapping $g\colon \rel C_k \to \rel C_3$ defined by
\(
  g(x) = f(x,\dots,x)
\).
Parts (1) and (2) of Lemma~\ref{lem:unary-degree} imply that $\deg g$ is an odd number not greater than $k/3$, consequently $\deg g \leq N$. Further, Lemma~\ref{lem:minor-preservation} implies that
\(
  \deg g = \sum_{i=1}^n \deg_i f
\),
and therefore we immediately get item (2). This argument also shows item (1) if $\deg_i f$ is non-negative for all $i$.  We reduce the general case to this case. More precisely, for each $f$ we will find a~new polymorphism $f'$ of the same arity such that $\mathopen|\deg_i f\mathclose| = \deg_i f'$ for all $i$.

This $f'$ is constructed by a~simple trick: `reversing' all coordinates $x_i$ with $\deg_i f$ negative.
Let us do that for one coordinate at a~time, and for simplicity, consider just the first coordinate. Let $\theta$ be the automorphism of $\rel C_k$ such that $\theta(0) = 0$ and $\theta(i) = k-i$ for $i\neq 0$, and define $f'$ as
\[
  f'(x_1,\dots,x_n) = f(\theta(x_1),x_2,\dots,x_n)
.\]
It is easy to see that $f'$ is also a~polymorphism.
We claim that $\deg_1 f' = -\deg_1 f$ and $\deg_i f' = \deg_i f$ for all $i\neq 1$. This follows from the fact that applying $\theta$ in the first coordinate reverses the orientation of all edges in $O_{k,1}^n$ and it does not change the orientation of edges in $O_{k,j}^n$ for $j\neq 1$. (This is clear since the orientation of edges in $O_{k,i}^n$ is given by the orientation in the $i$-th coordinate.) In other words, we have
\(
  f'_\e( O_{k,1}^n ) = f_\e( -O_{k,1}^n ) \) and \( f'_\e( O_{k,j}^n ) = f_\e( O_{k,j}^n )
\) for $j\neq 1$. This directly implies the claim about degrees.

Repeating this trick, we eventually obtain a~polymorphism whose degrees are all positive and up to the sign same as degrees of $f$. This completes the proof.
\end{proof}

Theorem~\ref{thm:key} now follows from Lemma~\ref{lem:minor-preservation} and Lemma~\ref{lem:bounding}.
 
\appendix
\section{A homomorphic equivalence of minions}
  \label{app:equivalence}

In this appendix, we prove that (for $k$ odd and $N$ the largest odd number smaller than or equal to $k/3$) the minion $\clo Z_{\leq N}$ is homomorphically equivalent to $\Pol(\rel C_k,\rel C_3)$, i.e., that there exists minion homomorphisms between these minions in both directions. A~minion homomorphism $\delta\colon \Pol(\rel C_k,\rel C_3)\to \clo Z_{\leq N}$ was given in Section~\ref{sec:proof}, here we provide one from  $\clo Z_{\leq N}$ to $\Pol(\rel C_k,\rel C_3)$.

\begin{lemma} \label{lem:app_b}
Let $k$ and $N$ be odd such that $N \leq k/3$. There exists a~minion homomorphism $\eta\colon \clo Z_{\leq N} \to \Pol(\rel C_k,\rel C_3)$.
\end{lemma}

\begin{proof}
For simplicity, let us assume that $k = 3N$. The general case is obtained by an~easy observation that $\Pol(\rel C_k,\rel C_3) \to \Pol(\rel C_{k'},\rel C_3)$ for any odd $k' \leq k$.
As an intermediate step, we define a~graph $\rel D_k$ with vertices $V(D_k) = V(C_k)$ such that $(u,v)\in E(D_k)$ if the distance of $u$ and $v$ in $\rel C_k$ is odd and at most $N$. Alternatively, we can say $(u,v)\in E(D_k)$ if they are connected in $\rel C_k$ by a~walk of length exactly $N$. It is easy to see that these two are equivalent. (See also Figure~\ref{fig:d9}.)
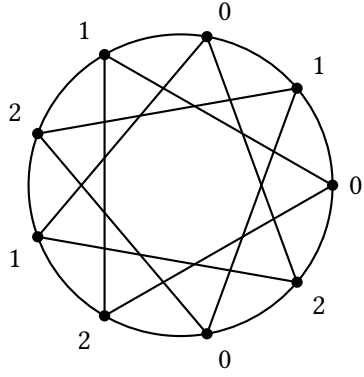
\begin{figure}
  \begin{centering}
  \begin{tikzpicture}[scale = 2]
    \draw [thick] (0,0) circle (1cm);
    \foreach \i/\c in {0/0,40/1,80/0,120/1,160/2,200/1,240/2,280/0,320/2} {
      \node [circle,fill,inner sep=1.5,label={\i:$\c$}] at (\i:1) {};
      \draw [thick] (\i:1) -- (120+\i:1);
    }
  \end{tikzpicture}
  \end{centering}
  \caption{The graph $\rel D_9$ with a~3-colouring $h_9$.}
    \label{fig:d9}
\end{figure}

We claim:
\begin{enumerate}
 \item there is a~minion homomorphism $\eta'$ from $\clo Z_{\leq N}$ to $\Pol(\rel C_k,\rel D_k)$, and
 \item there is a~graph homomorphism $h_k\colon \rel D_k \to \rel C_3$.
\end{enumerate}

To prove the first claim, we define $\eta'$ by
\[
  \eta'(f)\colon (x_1,\dots,x_n) \mapsto f(x_1,\dots,x_n) \bmod k
.\]
That is, we apply the function $f$ to the $n$-tuple of vertices of $\rel C_k$ as they would be numbers in $\mathbb Z$, and then take the residue modulo $k$ of the result to get a~number between $0$ and $k-1$. From the definition, it is clear that $\eta'$ is minor-preserving, therefore we only need to prove that each $\eta'(f)$ is a~polymorphism from $\rel C_k$ to $\rel D_k$. It is enough to prove the claim for $f$ of the form:
\(
  f(x_1,\dots,x_N) = \pm x_1 \pm \dots \pm x_N
\)
since all other functions in $\clo Z_{\leq N}$ are minors of some such $f$. Let $f' = \eta'(f)$, i.e.,
\[
  f'(x_1,\dots,x_N) = (\pm x_1 \pm \dots \pm x_N) \bmod k
\]
Assume that $(u_i,v_i) \in E(C_k)$ for $i = 1,\dots,N$, and observe that 
\[
  f'(u_1,\dots,u_N),f'(v_1,u_2,\dots,u_N),f'(v_1,v_2,u_3,\dots,u_N),
  \dots,f'(v_1,\dots,v_N)
\]
is a~walk in $\rel C_k$ from $f'(u_1,\dots,u_N)$ to $f'(v_1,\dots,v_N)$ of length exactly $N$, which implies that
\[
  (f'(u_1,\dots,u_N),f'(v_1,\dots,v_N))\in E(D_k)
\]
and $f' \in \Pol(\rel C_k,\rel D_k)$.

For the second claim, consider $h_k\colon V(D_k) \to V(C_3)$ defined by
\[
  h_k(x) = \begin{cases}
    0 & \text{ if $x < N$ and $x$ is oven, or
      $x > 2N$ and $x$ is odd, } \\
    1 & \text{ if $x < 2N$ and $x$ is odd, } \\
    2 & \text{ if $x > N$ and $x$ is even. }
  \end{cases}
\]
It is easy to check that $h$ is indeed a~graph homomorphism. (Figure~\ref{fig:d9} shows $h_9$.)

Finally, the minion homomorphism $\eta$ is given by
\[
  \eta(f)\colon (x_1,\dots,x_n) \mapsto h_k(\eta'(f)(x_1,\dots,x_n)).
\]
It is straight-forward to check that $\eta$ is minor-preserving, we also have that $\eta(f)$ is a~polymorphism from $\rel C_k$ to $\rel C_3$ since it is a~composition of a~polymorphism $\eta'(f)\colon \rel C_k^n \to \rel D_k$ and a~homomorphism $\rel D_k \to \rel C_3$.
\end{proof}

We remark, without giving a~proof, that the composition of minion homomorphisms $\delta\circ\eta$ is the identity on $\clo Z_{\leq N}$, i.e., for $f(x_1,\dots,x_n) = \sum_{i=1}^n c_ix_i$ and $i\in \{1,\dots,n\}$, we have $\deg_i \eta(f) = c_i$.

As a~simple corollary, we can obtain polymorphisms that forbid simple reductions from some other PCSPs.

\begin{corollary} Let $k\geq 9$ and $n\leq k/3$ be both odd, $\Pol(\rel C_k,\rel C_3)$ contains functions $c,s$, and $o$ satisfying:
\begin{enumerate}
  \item $c_n(x_1,\dots,x_n) \equals c(x_2,\dots,x_n,x_1)$,
  \item $s(x,y,x,z,y,z) \equals s(y,x,z,x,z,y)$, and
  \item $o(x,x,y,y,y,x) \equals o(x,y,x,y,x,y)\equals o(y,x,x,x,y,y)$.
\end{enumerate}
\end{corollary}

\begin{proof}
  Let $N$ be the largest odd number smaller than or equal to $k/3$, note that $n\leq N$ and $3\leq N$.
  Since a~minion homomorphism preserves satisfaction of the above identities, it is enough to find such functions in $\clo Z_{\leq N}$. For that, we consider
  \begin{align*}
    c_n(x_1,\dots,x_n) &= x_1 + \dots + x_n, \\
    s(x_1,\dots,x_6) &= x_1 + x_3 + x_5, \\
    o(x_1,\dots,x_6) &= x_1 + x_2 + x_3
  .\end{align*}
  Clearly, $c_n,s,o \in \clo Z_{\leq N}$, and it is easy to check that they satisfy the required identities. The claim is then given by Lemma~\ref{lem:app_b}.
\end{proof}

A~function satisfying item (3) of the above corollary is called an \emph{Ol\v{s}\'ak} function, and absence of such a~polymorphism is a~requirement for a~reduction from approximate hypergraph colouring using \cite[Corollary 6.3]{BKO18}. A~function satisfying item (2) is called a~\emph{Siggers} function, and absence of such shows that \cite[Theorem 3.1]{BKO18} cannot be used for a~reduction from approximate graph colouring to $\PCSP(\rel C_k,\rel C_3)$ (see also \cite[Theorem 6.9]{BKO18}).
 
\section*{Acknowledgements}

We would like to thank John Hunton for useful discussions about algebraic topology, and the reviewers for useful comments.

\bibliographystyle{alphaurl}

\begin{thebibliography}{BBKO19}

\bibitem[ABP18]{ABP18}
Per Austrin, Amey Bhangale, and Aditya Potukuchi.
\newblock Improved inapproximability of rainbow coloring.
\newblock arXiv:1810.02784, October 2018.
\newblock \url{http://arxiv.org/abs/1810.02784}.

\bibitem[ABP19]{ABP19}
Per Austrin, Amey Bhangale, and Aditya Potukuchi.
\newblock Simplified inapproximability of hypergraph coloring via $t$-agreeing
  families.
\newblock arXiv:1904.01163, April 2019.
\newblock \url{http://arxiv.org/abs/1904.01163}.

\bibitem[AGH17]{AGH17}
Per Austrin, Venkatesan Guruswami, and Johan H{\aa}stad.
\newblock (2+{$\epsilon$})-{S}at is {NP}-hard.
\newblock {\em {SIAM} J. Comput.}, 46(5):1554--1573, 2017.
\newblock \href {https://doi.org/10.1137/15M1006507} {doi:10.1137/15M1006507}.

\bibitem[BBKO19]{BKO18}
Libor Barto, Jakub Bul{\'i}n, Andrei Krokhin, and Jakub Opr{\v{s}}al.
\newblock Algebraic approach to promise constraint satisfaction.
\newblock arXiv:1811.00970, version 3, June 2019.
\newblock \url{https://arxiv.org/abs/1811.00970v3}.

\bibitem[BG16]{BG16-graph}
Joshua Brakensiek and Venkatesan Guruswami.
\newblock New hardness results for graph and hypergraph colorings.
\newblock In Ran Raz, editor, {\em 31st Conference on Computational Complexity
  (CCC 2016)}, volume~50 of {\em Leibniz International Proceedings in
  Informatics (LIPIcs)}, pages 14:1--14:27, Dagstuhl, Germany, 2016. Schloss
  Dagstuhl--Leibniz-Zentrum fuer Informatik.
\newblock \href {https://doi.org/10.4230/LIPIcs.CCC.2016.14}
  {doi:10.4230/LIPIcs.CCC.2016.14}.

\bibitem[BG17]{BG17}
Joshua Brakensiek and Venkatesan Guruswami.
\newblock The quest for strong inapproximability results with perfect
  completeness.
\newblock In {\em Approximation, Randomization, and Combinatorial Optimization.
  Algorithms and Techniques, {APPROX/RANDOM} 2017, August 16--18, 2017,
  Berkeley, CA, {USA}}, pages 4:1--4:20, 2017.
\newblock \href {https://doi.org/10.4230/LIPIcs.APPROX-RANDOM.2017.4}
  {doi:10.4230/LIPIcs.APPROX-RANDOM.2017.4}.

\bibitem[BG18]{BG18-structure}
Joshua Brakensiek and Venkatesan Guruswami.
\newblock Promise constraint satisfaction: Structure theory and a symmetric
  boolean dichotomy.
\newblock In {\em Proceedings of the Twenty-Ninth Annual ACM-SIAM Symposium on
  Discrete Algorithms}, SODA'18, pages 1782--1801, Philadelphia, PA, USA, 2018.
  Society for Industrial and Applied Mathematics.

\bibitem[BG19]{BG19}
Joshua Brakensiek and Venkatesan Guruswami.
\newblock An algorithmic blend of {LP}s and ring equations for promise {CSP}s.
\newblock In {\em Proceedings of the Thirtieth Annual {ACM-SIAM} Symposium on
  Discrete Algorithms, {SODA} 2019, San Diego, California, USA, January 6-9,
  2019}, pages 436--455, 2019.
\newblock \href {https://doi.org/10.1137/1.9781611975482.28}
  {doi:10.1137/1.9781611975482.28}.

\bibitem[Bha18]{Bha18}
Amey Bhangale.
\newblock {NP}-hardness of coloring 2-colorable hypergraph with
  poly-logarithmically many colors.
\newblock In {\em 45th International Colloquium on Automata, Languages, and
  Programming, {ICALP} 2018, July 9-13, 2018, Prague, Czech Republic}, pages
  15:1--15:11, 2018.

\bibitem[BK16]{BK16}
Libor Barto and Marcin Kozik.
\newblock Robustly solvable constraint satisfaction problems.
\newblock {\em SIAM Journal on Computing}, 45(4):1646--1669, 2016.
\newblock \href {https://doi.org/10.1137/130915479} {doi:10.1137/130915479}.

\bibitem[BKO19]{BKO19}
Jakub Bul{\'i}n, Andrei Krokhin, and Jakub Opr{\v{s}}al.
\newblock Algebraic approach to promise constraint satisfaction.
\newblock In {\em Proceedings of the 51st Annual ACM SIGACT Symposium on the
  Theory of Computing (STOC ’19)}, New York, NY, USA, 2019. ACM.
\newblock \href {https://doi.org/10.1145/3313276.3316300}
  {doi:10.1145/3313276.3316300}.

\bibitem[BKW17]{BKW17}
Libor Barto, Andrei Krokhin, and Ross Willard.
\newblock Polymorphisms, and how to use them.
\newblock In Andrei Krokhin and Stanislav {\v Z}ivn{\' y}, editors, {\em The
  Constraint Satisfaction Problem: Complexity and Approximability}, volume~7 of
  {\em Dagstuhl Follow-Ups}, pages 1--44. Schloss Dagstuhl--Leibniz-Zentrum
  fuer Informatik, Dagstuhl, Germany, 2017.
\newblock \href {https://doi.org/10.4230/DFU.Vol7.15301.1}
  {doi:10.4230/DFU.Vol7.15301.1}.

\bibitem[Bul17]{Bul17}
Andrei~A. Bulatov.
\newblock A dichotomy theorem for nonuniform {CSP}s.
\newblock In {\em 2017 IEEE 58th Annual Symposium on Foundations of Computer
  Science (FOCS)}, pages 319--330, October 2017.
\newblock \href {https://doi.org/10.1109/FOCS.2017.37}
  {doi:10.1109/FOCS.2017.37}.

\bibitem[DMR09]{DMR09}
Irit Dinur, Elchanan Mossel, and Oded Regev.
\newblock Conditional hardness for approximate coloring.
\newblock {\em {SIAM} J. Comput.}, 39(3):843--873, 2009.
\newblock \href {https://doi.org/10.1137/07068062X} {doi:10.1137/07068062X}.

\bibitem[DRS05]{DRS05}
Irit Dinur, Oded Regev, and Clifford Smyth.
\newblock The hardness of 3-uniform hypergraph coloring.
\newblock {\em Combinatorica}, 25(5):519--535, September 2005.
\newblock \href {https://doi.org/10.1007/s00493-005-0032-4}
  {doi:10.1007/s00493-005-0032-4}.

\bibitem[FV98]{FV98}
Tom{\'{a}}s Feder and Moshe~Y. Vardi.
\newblock The computational structure of monotone monadic {SNP} and constraint
  satisfaction: A study through datalog and group theory.
\newblock {\em SIAM J. Comput.}, 28(1):57--104, February 1998.
\newblock \href {https://doi.org/10.1137/S0097539794266766}
  {doi:10.1137/S0097539794266766}.

\bibitem[GJ76]{GJ76}
M.~R. Garey and David~S. Johnson.
\newblock The complexity of near-optimal graph coloring.
\newblock {\em J. {ACM}}, 23(1):43--49, 1976.
\newblock \href {https://doi.org/10.1145/321921.321926}
  {doi:10.1145/321921.321926}.

\bibitem[GK04]{GK04}
Venkatesan Guruswami and Sanjeev Khanna.
\newblock On the hardness of 4-coloring a 3-colorable graph.
\newblock {\em SIAM Journal on Discrete Mathematics}, 18(1):30--40, 2004.
\newblock \href {https://doi.org/10.1137/S0895480100376794}
  {doi:10.1137/S0895480100376794}.

\bibitem[GL74]{GL74}
Donald~L. Greenwell and L{\'a}szl{\'o} Lov{\'a}sz.
\newblock Applications of product colouring.
\newblock {\em Acta Mathematica Academiae Scientiarum Hungarica},
  25(3):335--340, September 1974.
\newblock \href {https://doi.org/10.1007/BF01886093} {doi:10.1007/BF01886093}.

\bibitem[GL17]{GL17}
Venkatesan Guruswami and Euiwoong Lee.
\newblock Strong inapproximability results on balanced rainbow-colorable
  hypergraphs.
\newblock {\em Combinatorica}, December 2017.
\newblock \href {https://doi.org/10.1007/s00493-016-3383-0}
  {doi:10.1007/s00493-016-3383-0}.

\bibitem[Hat01]{Hat01}
Allen Hatcher.
\newblock {\em Algebraic topology}.
\newblock Cambridge University Press, Cambridge, 2001.

\bibitem[HN90]{HN90}
Pavol Hell and Jaroslav Ne{\v{s}}et{\v{r}}il.
\newblock On the complexity of {$H$}-coloring.
\newblock {\em J. Combin. Theory Ser. B}, 48(1):92--110, 1990.

\bibitem[HN04]{HN04}
Pavol Hell and Jaroslav Ne\v{s}et\v{r}il.
\newblock {\em Graphs and Homomorphisms}.
\newblock Oxford University Press, 2004.

\bibitem[Hua13]{Hua13}
Sangxia Huang.
\newblock Improved hardness of approximating chromatic number.
\newblock In Prasad Raghavendra, Sofya Raskhodnikova, Klaus Jansen, and
  Jos{\'e} D.~P. Rolim, editors, {\em Approximation, Randomization, and
  Combinatorial Optimization. Algorithms and Techniques: 16th International
  Workshop, APPROX 2013, and 17th International Workshop, RANDOM 2013,
  Berkeley, CA, USA, August 21-23, 2013. Proceedings}, pages 233--243, Berlin,
  Heidelberg, 2013. Springer.
\newblock \href {https://doi.org/10.1007/978-3-642-40328-6{\_}17}
  {doi:10.1007/978-3-642-40328-6{\_}17}.

\bibitem[KLS00]{KLS00}
Sanjeev Khanna, Nathan Linial, and Shmuel Safra.
\newblock On the hardness of approximating the chromatic number.
\newblock {\em Combinatorica}, 20(3):393--415, Mar 2000.
\newblock \href {https://doi.org/10.1007/s004930070013}
  {doi:10.1007/s004930070013}.

\bibitem[KT17]{KT17}
Ken{-}ichi Kawarabayashi and Mikkel Thorup.
\newblock Coloring 3-colorable graphs with less than {$n^{1/5}$} colors.
\newblock {\em J. {ACM}}, 64(1):4:1--4:23, 2017.

\bibitem[K{\v{Z}}17]{KZ17}
Andrei Krokhin and Stanislav {\v{Z}}ivn\'y, editors.
\newblock {\em The Constraint Satisfaction Problem: Complexity and
  Approximability}, volume~7 of {\em Dagstuhl Follow-Ups}.
\newblock Schloss Dagstuhl -- Leibniz-Zentrum f\"ur Informatik, 2017.

\bibitem[Lov78]{Lov78}
L{\'a}szl{\'o} Lov{\'a}sz.
\newblock Kneser's conjecture, chromatic number, and homotopy.
\newblock {\em Journal of Combinatorial Theory, Series A}, 25(3):319--324,
  1978.

\bibitem[Mat03]{Mat03}
Ji\v{r}\'i Matou\v{s}ek.
\newblock {\em Using the Borsuk-Ulam Theorem}.
\newblock Lectures on Topological Methods in Combinatorics and Geometry.
  Springer-Verlag Berlin Heidelberg, first edition, 2003.
\newblock \href {https://doi.org/10.1007/978-3-540-76649-0}
  {doi:10.1007/978-3-540-76649-0}.

\bibitem[Wal83]{Wal83-equivariant}
James~W. Walker.
\newblock From graphs to ortholattices and equivariant maps.
\newblock {\em Journal of Combinatorial Theory, Series B}, 35(2):171--192,
  1983.

\bibitem[W{\v{Z}}19]{WZ19}
Marcin Wrochna and Stanislav {\v{Z}}ivn{\'y}.
\newblock Improved hardness for {$H$}-colourings of {$G$}-colourable graphs.
\newblock arXiv:1907.00872, July 2019.
\newblock \url{http://arxiv.org/abs/1907.00872}.

\bibitem[Zhu17]{Zhu17}
Dmitriy Zhuk.
\newblock A proof of {CSP} dichotomy conjecture.
\newblock In {\em 2017 IEEE 58th Annual Symposium on Foundations of Computer
  Science (FOCS)}, pages 331--342, Oct 2017.
\newblock \href {https://doi.org/10.1109/FOCS.2017.38}
  {doi:10.1109/FOCS.2017.38}.

\end{thebibliography}

\end{document}